\documentclass{article}
\usepackage{amsmath,amsthm,amssymb}
\usepackage{url}
\usepackage{listings}
\usepackage{enumerate}
\usepackage{makecell}
\usepackage{graphicx}
\usepackage[all]{xy}
\usepackage{mathtools} 
\usepackage{url}
\usepackage{array,multirow}
\usepackage{amsfonts}
\usepackage{tikz}
\usepackage{booktabs} 
\usetikzlibrary{matrix}
\usepackage{lscape}
\usetikzlibrary{automata,positioning}
\usepackage{caption}
\usepackage{subcaption}
\usepackage{multicol}
\usepackage{mleftright}
\usepackage{float}
\usepackage{chngpage}
\usetikzlibrary{decorations.pathmorphing}
\tikzset{snake it/.style={decorate, decoration=snake}}
\usepackage{hyperref}
\usepackage{pgf}
\usepackage{cleveref}

\newtheorem{theorem}{Theorem}[section]
\newtheorem{proposition}[theorem]{Proposition}
\newtheorem{lemma}[theorem]{Lemma}

\theoremstyle{definition}
\newtheorem{definition}[theorem]{Definition}

\newtheorem{example}[theorem]{Example}
\newtheorem{remark}[theorem]{Remark}

\newcommand{\abs}[1]{|#1|}
\newcommand{\car}[1]{\#\left(#1\right)}

\newcommand{\lam}{\lambda}
\newcommand{\maxsc}{\mathsf{maxsc}}
\renewcommand{\sc}{\mathsf{sc}}
\DeclareMathOperator{\ran}{ran}
\DeclareMathOperator{\dom}{dom}
\begin{document}
	\title{The number of languages with \\ maximum state complexity}
	\author{Bj\o rn Kjos-Hanssen\thanks{This work was partially supported by grants from the
		Simons Foundation (\#315188 and \#704836 to Bj\o rn Kjos-Hanssen) and
		Decision Research Corporation (University of Hawai\textquoteleft i Foundation Account \#129-4770-4). We are grateful to the gracious referee who persisted through seven revisions of the paper.}
	\\ Lei Liu}
	\maketitle

\begin{abstract}
C\^{a}mpeanu and Ho (2004) determined the maximum finite state complexity of finite languages,
building on work of Champarnaud and Pin (1989).
They stated that it is very difficult to determine the number of maximum-complexity languages.
Here we give a formula for this number. We also generalize their work from languages to functions on finite sets.
\end{abstract}

	\section{Introduction}
		At some point in the 1980s, Howard Straubing posed a problem that was subsequently solved in Champarnaud and Pin (1989)~\cite{CHAMPARNAUD198991}.
		They showed that the minimal incomplete deterministic finite automaton of a language $L\subseteq\Sigma^n$, where $\Sigma=\{0,1\}$, has at most
		\[
			\sum_{i=0}^n \min(2^i, 2^{2^{n-i}}-1)
		\]
		states. Moreover, for each $n$ there exists an $L$ attaining this bound.
		C\^{a}mpeanu and Ho (2004)~\cite{MR2167778} showed more generally
		that the tight upper bound for $\Sigma$ of cardinality $k$ and for complete automata is
		\[
			\frac{k^r-1}{k-1} + \sum_{j=0}^{n-r}(2^{k^j}-1) + 1
		\]
		where $r=\min\{m:k^m\ge 2^{k^{n-m}}-1\}$.
		(In these results, requiring totality of the transition function adds 1 to the state count.)
		C\^{a}mpeanu and Ho's result can be viewed as concerning functions $f:[k]^n\to [2]$ where $[k]=\{0,\dots,k-1\}$ is a set of cardinality $k$.
		We generalize their result to arbitrary functions $f:[k]^n\to [c]$ where $c$ is a positive integer.
		Equivalently, we consider functions $f:[k]^*\to [c]$, where $\{x:f(x)>0\}\subseteq [k]^n$ for some $n$,
		and where automata have $c-1$ accept states corresponding to nonzero values of $f$.

		The function $+$ on $\mathbb Z/5\mathbb Z$ may seem rather complicated as functions on that set go.
		On the other hand, $f(x,y,z)=x+y+z$ mod 5 is less so, in that we can decompose it as $(x+y)+z$, so that after seeing $x$ and $y$,
		we need not remember the pair $(x,y)$, but only their sum. Out of the $5^{5^3}$ ternary functions on a 5-element set,
		at most $5^{2\cdot 5^2}$ can be decomposed as $(x*_1y)*_2 z$ for some binary functions $*_1$, $*_2$.
		This idea of the state complexity of functions has been applied in bioinformatics \cite{PolErs19}.
		In \Cref{sec:main} we make precise a sense in which such functions are not the most complex ternary functions.
		We do this by extending a result of C\^ampeanu and Ho~\cite{MR2167778} to functions taking values in a set of size larger than two.
		Rising to an implicit challenge posed by C\^ampeanu and Ho, we give a formula for the number of maximally complex languages.

		The structure of the paper is as follows. In \Cref{sec:main} we obtain an upper bound in \Cref{thm:ub} for the complexity of a function $f:[b]^n\to [c]$, and
		a matching lower bound in \Cref{thm:lb}.
		In \Cref{sec:number} we obtain the number of maximal complexity functions in \Cref{thm:jan16-2021}.
		Then we look at asymptotics in \Cref{sub:sur}, culminating in \Cref{thm:nmcf}.

	\section{Complexity of languages and operations}\label{sec:main}
		Let $\lam$ denote the empty word.
		Let the cardinality of a finite set $A$ be denoted by $\#(A)$, and the length of a finite word $w$ by $\abs{w}$.
		We define a function $\mathbb{I}_A:B\to A\cup\{0\}$ for any sets $A\subseteq B$ with $0\not\in A$ by
		\[
			\mathbb{I}_A(x)=\begin{cases}x &\text{if }x\in A,\\ 0 &\text{if }x\not\in A.\end{cases}
		\]
		\begin{definition}\label{reviewer2}\label{gener}
			Let $b$ and $c$ be positive integers and let $\Sigma$ be an alphabet with $\car{\Sigma}=b$.
			An \emph{incomplete deterministic finite automaton} (IDFA)
			$M$ is a 5-tuple $(Q, \Sigma, \delta, q_0, F)$, where
			$Q$ is a finite set of states,
			$\Sigma$ is a finite alphabet,
			$q_0 \in Q$ is the start state,
			$F\subseteq Q$ is the set of accept states,
			and
			$\delta : D \to Q$, where $D\subseteq Q\times\Sigma$, is the transition function.

			W also require $F=\{1,\dots,c-1\}=[c]\setminus\{0\}$, where $c-1=\car{F}$.
			If $D=Q\times\Sigma$, i.e., $\delta$ is total, then $M$ is moreover a \emph{deterministic finite automaton} (DFA).

			We define $\overline{\delta}:D\to Q$, where $D\subseteq Q\times\Sigma^*$,
			by $\overline{\delta}(q,\lam)=q$, and recursively $\overline{\delta}(q,x u) = \delta(\overline{\delta}(q,x),u)$ for $x\in\Sigma^*$ and $u\in\Sigma$.
			We say that states $q_1,q_2$ are \emph{$M$-distinguishable} if there is a $z$ with $\overline\delta(q_1,z)\ne \overline\delta(q_2,z)$ and 
			$\{\overline\delta(q_1,z),\overline\delta(q_2,z)\}\cap F\ne\emptyset$.

			The \emph{function accepted by $M$} is the function $f:\Sigma^*\to [c]$ defined by
			\[
			f(x) = \mathbb{I}_F(\overline{\delta}(q_0,x)),\quad \text{if $\overline{\delta}(q_0,x)$ is defined},
			\]
			and $f(x)=0$ otherwise.
			Thus $f(x)=0$ if $\overline{\delta}(q_0,x)\not\in F$, and $f(x)=\overline{\delta}(q_0,x)$ if $\overline{\delta}(q_0,x)\in F$.
			The \emph{language accepted by} $M$ is
			\[
				L(M)=\{x\in\Sigma^*:f(x)>0\}=\{x\in\Sigma^*:\overline\delta(q_0,x)\in F\}.
			\]
		\end{definition}
		Note that in the case $c=2$, accepting a language is equivalent to accepting its indicator (characteristic) function.
		\begin{definition}[state complexity]\label{df:sc}
			We call an IDFA $M=(Q,\Sigma,\delta,q_0,F)$ \emph{minimal (for $L(M)$)} if $\car{Q}\le\car{Q'}$
			for all IDFAs $M'=(Q',\Sigma,\delta',q_0',F')$ with $L(M')=L(M)$.
			Moreover, $M$ is \emph{minimal for $f$} if $M$ accepts $f$ and $\car{Q}\le\car{Q'}$ for all $M'$ accepting $f$.
			In this case we define the state complexity $\sc(f)$ by $\sc(f)=\car{Q}$.
		\end{definition}

		Champarnaud and Pin~\cite{CHAMPARNAUD198991} obtained the following result.
		\begin{theorem}[{\cite[Theorem 4]{CHAMPARNAUD198991}}]\label{thm:CP}
			A minimal IDFA for a language $L\subseteq\{0,1\}^n$ has at most
			\[
				\sum_{i=0}^n \min(2^i, 2^{2^{n-i}}-1)
			\]
			states, and for each $n$ there exists a language $L$ attaining this bound.
		\end{theorem}
		\Cref{thm:CP} was generalized by C\^ampeanu and Ho~\cite{MR2167778}:
		\begin{theorem}[{\cite[Corollary 10]{MR2167778}}]
			Let $k\ge 1$ and $l\ge 0$ be integers, and let $M$ be a minimal DFA for a language $L\subseteq [k]^l$.
			Let $Q$ be the set of states of $M$. Then we have:
			\begin{enumerate}[(i)]
				\item\label{i} $\car{Q}\le \frac{k^r-1}{k-1}+\sum_{j=0}^{l-r} (2^{k^j}-1)+1$, where $r=\min\{m\mid k^m\ge 2^{k^{l-m}}-1\}$.
				\item There is an $M$ such that the upper bound given by \Cref{i} is attained.
			\end{enumerate}
		\end{theorem}
		Both of these results involve an upper bound which can be viewed as a special case of \Cref{thm:ub} below.

		We now develop a function version of the Myhill--Nerode theorem, by following and generalizing the presentation in Shallit's textbook~\cite{Shallit:2008:SCF:1434864}.
		\begin{definition}
			Let $\Sigma$ be an alphabet and let $c\in\mathbb N$.
			A relation $R\subseteq \Sigma^*\times\Sigma^*$ is \emph{right invariant} if for all $x,y,z\in\Sigma^*$, we have
			$x R y\implies xz R yz$.
			An equivalence relation $E$ on $\Sigma^*$ is a \emph{congruence relation} for $f:\Sigma^*\to [c]$ if for all $x,y\in\Sigma^*$,
			\(
				xEy\implies f(x)=f(y).
			\)
			For an equivalence relation $E$, the \emph{index} of $E$, denoted $\mathrm{index}(E)$, is the number of equivalence classes of $E$.
			An equivalence relation has \emph{finite index} if $\mathrm{index}(E)<\infty$.
			The \emph{Myhill--Nerode equivalence relation} for $f:\Sigma^*\to [c]$ is the relation $R_f$ defined by
			\[
				x R_f y\iff \text{for all }z\in\Sigma^*, f(xz)=f(yz).
			\]
			Let $[x]_{f}$ denote the $R_{f}$-equivalence class of $x$.
		\end{definition}
		\begin{lemma}
			Let $f:\Sigma^*\to [c]$.
			\begin{enumerate}
				\item\label{R_f i} $R_f$ is an equivalence relation.
				\item\label{R_f ii} $R_f$ right invariant.
			\end{enumerate}
		\end{lemma}
		\begin{proof}
			\Cref{R_f i} is a standard observation. For \Cref{R_f ii}:
			If we extend $xz$ and $yz$ by the same string $w$,
			then we have also extended $x$ and $y$ by the same string $zw$, and hence $f(xzw)=f(yzw)$.
		\end{proof}
		\begin{lemma}\label{lem:shallit393}
			Let $f:\Sigma^*\to [c]$.
			Suppose that $E$ is a right invariant equivalence relation on $\Sigma^*$ which is a congruence relation for $f$.
			Then $E$ is a refinement of $R_{f}$.
		\end{lemma}
		\begin{proof}
			We must show that $xEy\implies x R_{f} y$.
			Suppose $xEy$ and let $z\in\Sigma^*$. Since $E$ is right invariant, $xz E yz$.
			Since $E$ is a congruence relation for $f$, $f(xz)=f(yz)$. Thus we have shown that $xR_{f} y$.
		\end{proof}
		Every function is onto its range, and when the range is a finite subset of $\mathbb N$, when studying complexity under our definitions we
		assume the range is an initial segment of $\mathbb N$. Thus we restrict attention to onto functions in \Cref{thm:myhill}.
		\begin{theorem}\label{thm:myhill}
			Let $f:\Sigma^*\to [c]$ be onto.
			The following are equivalent:
			\begin{enumerate}
				\item $f$ is accepted by some IDFA.
				\item There exists a right invariant congruence relation for $f$ of finite index.
				\item $R_f$ has finite index.
				\item $f$ is accepted by some DFA.
			\end{enumerate}
		\end{theorem}
		\begin{proof}
			We prove this in the usual round-robin fashion.
				\begin{description}
				\item[(1) $\implies$ (2):] Let $M$ be an IDFA that accepts $f$. Define a relation $R_M$ by $xR_M y$ iff $\overline\delta(q_0,x)=\overline\delta(q_0,y)$,
				or both are undefined.
				Since $M$ has finitely many states, $R_M$ has finite index. From the definition of $\overline\delta$ it follows that $R_M$ is right invariant.
				Finally, since $f(x)=\mathbb{I}_F(\overline\delta(q_0,x))$ if defined, and 0 otherwise, $f(x)$ is determined by $\overline\delta(q_0,x)$.
				Thus $R_M$ is a congruence relation for $f$.
				\item[(2) $\implies$ (3):] Let $R$ be a right invariant congruence relation for $f$, of finite index.
				By \Cref{lem:shallit393}, $R$ is a refinement of $R_f$.
				Then $\mathrm{index}(R_f)\le \mathrm{index}(R)<\infty$, as desired.
				\item[(3) $\implies$ (4):] Suppose $R_f$ has finite index. Define
				$Q'=\{[x]_f:x\in\Sigma^*\}$,
				$q_0'=[\lam]_f$,
				$F'=\{[x]_f: f(x)>0\}$, and
				$\delta'([x]_f,a)=[xa]_f$.
				Then $\#(Q')=\mathrm{index}(R_f)<\infty$.
				Since $R_f$ is right invariant, $\delta'$ is well-defined.
				Thus $M'=(Q',\Sigma,\delta',q_0',F')$ is an IDFA.\@ 
				We must show that $f(x)=\mathbb{I}_{F'}(\overline{\delta'}(q_0',x))$ for each $x$.
				Case 1: $f(x)=0$. Since $R_f$ is a congruence relation for $f$, $[x]_f\not\in F'$ and hence
				$\overline{\delta'}(q_0',x)=[\lambda x]_f=[x]_f\not\in F'$ which means that $\mathbb{I}_{F'}(\overline{\delta'}(q_0',x))=0$.
				Case 2: $f(x)>0$. Then by definition $[x]_f\in F'$ and so $\overline{\delta'}(q_0',x)=[\lambda x]_f=[x]_f\in F'$ 
				which means that $\mathbb{I}_{F'}(\overline{\delta'}(q_0',x))=\overline{\delta'}(q_0',x)$.
				Finally, let $\pi:F'\to [c]\setminus\{0\}$ be a bijection and formally replace each $q\in F'$ by $\pi(q)\in [c]$.
				\item[(4) $\implies$ (1):] This is immediate since each DFA is an IDFA.
			\end{description}
		\end{proof}

		\begin{theorem}\label{thm:shallit3.10.1}\label{thm:blue}
			Let $f:\Sigma^*\to [c]$.
			Let $M$ be an IDFA accepting $f$. Let $q$ be the number of states of $M$.
			Suppose that all states of $M$ are reachable and that any two states of $M$ are $M$-distinguishable.
			Then $\sc(f)=q$.
		\end{theorem}
		\begin{proof}			
			Let $M'$ be the automaton in \Cref{thm:myhill} for $f$ and let $Q'$ be its set of states. We claim that $M'$ is minimal.
			Note that $\#(Q')=\mathrm{index}(R_f)$.
			Let $N$ be any automaton accepting $f$, let $Q$ be its set of states and $\delta$ its transition function.
			Since $N$ accepts $f$, for all $x,y,z$,
			if $f(xz)\ne f(yz)$ then $\overline\delta(q_0,x)\ne \overline\delta(q_0,y)$.
			Thus $[x]_f\mapsto \overline\delta(q_0,x)$ is injective, and we have established that $\mathrm{index}(R_f)\le\#(Q)$, and hence that $M'$ is minimal.

			Now let $M=(Q,\Sigma,\delta,q_0,F)$ be any IDFA accepting $f$ for which any two states are reachable and $M$-distinguishable.
			It suffices to show that $\car{Q}\le\car{Q'}$, and for this it suffices to
			give an injective map $\varphi:Q\to Q'$.
			For each $q\in Q$ we let
			\begin{equation}\label{well}
				\varphi(q) = [x]_f\quad\text{where $x$ is such that $\overline\delta(q_0,x)=q$.}
			\end{equation}
			Such an $x$ must exist, or else $q$ is not reachable.

			\noindent\emph{Claim:} $\varphi$ is well-defined by~\eqref{well}.
			\begin{proof}[Proof of claim] Suppose that $\overline\delta(q_0,y)=q$ and let us show $[x]_f=[y]_f$.
			Let $z\in\Sigma^*$.
			Since $M$ accepts $f$,
			\begin{itemize}
				\item for $i>0$,
				$f(xz)=i$ iff $\delta(q_0,xz)=i$ and
				$f(yz)=i$ iff $\delta(q_0,yz)=i$; and
				\item for $i=0$,
				$f(xz)=0$ iff $\delta(q_0,xz)$ is undefined or is not in $F$, and 
				$f(yz)=0$ iff $\delta(q_0,yz)$ is undefined or is not in $F$
			\end{itemize}
			We have
			\[
				\delta(q_0,xz)=\delta(\delta(q_0,x),z)=\delta(q,z)
				=\delta(\delta(q_0,y),z)=\delta(q_0,yz)
			\]
			in the sense that $\delta(q_0,xz)$ and $\delta(q_0,yz)$
			 are both definitionally equal to $\delta(q,z)$, which may or may not be defined or in $F$.
			So in all cases $f(xz)=f(yz)$. \end{proof}

			Finally, let us show that $\varphi$ is one-to-one. If $\varphi(q_1)=\varphi(q_2)$ then $[x_1]_f=[x_2]_f$
			where $\delta(q_0,x_i)=q_i$. We will show, using $M$-distinguishability, that $q_1=q_2$.

			Suppose $q_1\ne q_2$. Then there is some $z$ with
			\[
				\delta(q_0,x_1z)=\delta(q_1,z)\ne \delta(q_2,z)=\delta(q_0,x_2z)
			\]
			and
			\(
				\{\delta(q_0,x_1z), \delta(q_0,x_2z)\}\cap F\ne\emptyset.
			\)
			Hence since $M$ accepts $f$, $f(x_1z)\ne f(x_2 z)$, which contradicts $[x_1]_f=[x_2]_f$.
		\end{proof}

		We write $A^B$ for the set of all functions from $B$ to $A$.
		\begin{definition}
			Let $b$ and $c$ be positive integers and let ${[c]}^{{[b]}^n}$ be the set of $n$-ary functions $f:[b]^n\to [c]$.
			Let $\mathfrak C\subseteq {[c]}^{{[b]}^n}$.
			The \emph{Champarnaud--Pin family} of $\mathfrak C$ is the family of sets $\{\mathfrak C_k\}_{0\le k\le n}$, where
			$\mathfrak C_k\subseteq {[c]}^{{[b]}^{n-k}}$, $0\le k\le n$,
			given by
			\[
				\mathfrak C_k = \{g\in [c]^{[b]^{n-k}} : \exists f\in\mathfrak C,\, w\in [b]^k\quad \forall x\quad g(x)=f(wx)\}.
			\]
			In terms of the function $\tau_w(x)=wx$, this can be restated as
			\[
				\mathfrak C_k = \{f\circ\tau_w \in [c]^{[b]^{n-k}} : f\in\mathfrak C,\, w\in [b]^k\}.
			\]
		\end{definition}
		So $\mathfrak C_0=\mathfrak C$, $\mathfrak C_1$ is obtained from $\mathfrak C_0$ by plugging in constants for the first input, and so forth.
		We write $\mathfrak C_n^-=\{f\in\mathfrak C_n: f(x)>0\text{ for some $x$}\}$.
		Note that $\car{\mathfrak C_n^-}\ge\car{\mathfrak C_n}-1$.

		\begin{definition}\label{dec17-21}
			Let us say that an IDFA $M$ accepts $f:[b]^n\to [c]$ if $M$ accepts the function $f^+:\Sigma^*\to [c]$ with $f^+(x)=f(x)$ if $x\in [b]^n$, and $f^+(x)=0$ otherwise.
			The \emph{state complexity} of $f:[b]^n\to [c]$ is the minimum number of states of an IDFA accepting $f:[b]^n\to [c]$, and is denoted $\sc(f)$.
		\end{definition}
		Note that \Cref{dec17-21} says that $\sc(f)=\sc(f^+)$.
		For $c>b=2$, $\sc(f)$ corresponds to automatic complexity of equivalence relations on binary strings as studied in~\cite{MR3712310}.
		The case $b=c$ is that of $n$-ary operations on a given finite set, which is of interest in universal algebra.
		
		We also define $\maxsc_{b,c,n}=\sum_{i=0}^n \min(b^i,c^{b^{n-i}}-1)$, which shall turn out to be the maximum of $\sc(f)$ over all $f$.

		\begin{definition}\label{def:crossover}
		We define a \emph{crossover function} $\chi(b,c,n)=\max\{i\in [0,n]\mid b^i\le c^{b^{n-i}}-1\}$.
		\end{definition}

		\begin{definition}\label{3pm-121721}
			Let $f\in [c]^{[b]^n}$ and $0\le j\le n$.
			We define an IDFA $M_{f,j}$.
			Its set of states is the disjoint union
			\[
				Q = \{q_w: w\in [b]^i, i\le j\} \cup
				\{r_{g}: g\in\mathfrak C^-_i, i>j\}.
			\]
			where all $q_w$, $r_g$ are distinct.
			The transition function $\delta$ of $M_{f,j}$ is given by
			\begin{eqnarray*}
				\delta(q_w,a) &=& \begin{cases}
				q_{wa}				& \abs{w}<j, a\in [b], \\
				r_{f\circ\tau_{wa}}	& \abs{w}=j, \text{$f\circ\tau_{wa}\not\equiv 0$},
				\end{cases}\\
				 \delta(r_g,a)	&=&r_{g\circ\tau_a},	\quad g\circ\tau_a\not\equiv 0.
			\end{eqnarray*}
		\end{definition}
		\begin{theorem}\label{thm:ub}
			Let $b$ and $c$ be positive integers.
			Let $f\in [c]^{[b]^n}$.
			Then
			\(
				\sc(f)\le \maxsc_{b,c,n}.
			\)
		\end{theorem}
		\begin{proof}
			Let $f\in\mathfrak C$. We must show that there is an
			IDFA $M_f$ accepting $f$ with at most the given number of states.
			Let $i_0=\chi(b,c,n)$ and let $M_f=M_{f,i_0}$ (\Cref{3pm-121721}).
			Then
			$\min(b^i,\car{\mathfrak C^-_{i}})=b^i$ for $i\le i_0$ and $\min(b^i,\car{\mathfrak C^-_{i}})=\car{\mathfrak C^-_i}$ for $i>i_0$.
			Note that for each $q\in Q$ there is an integer $i(q)$ such that $i(q)\le i_0\implies q=q_w$ for some $w$ and $i(q)>i_0\implies q=r_g$ for some $g$.
			The transition function $\delta$ is given by \Cref{3pm-121721} and also described in \Cref{fig:define-delta}.
			Note that if $b^n + 1 < c$, we may not have $i_0 \le n$,
			but this is ruled out because then no $f:[b]^n\to [c]$ can be onto (\Cref{def:crossover}).
			(We may assume that $f$ is onto, since otherwise a smaller IDFA can be found.)

			\begin{figure}
			\begin{tabular}{l  l  l  l}
				\toprule
			Cases	& $b^i <c^{b^{n-i}}-1$	& \makecell[l]{$b^i <c^{b^{n-i}}-1$ and \\ $b^{i+1}\ge c^{b^{n-(i+1)}}-1$} & $b^i\ge c^{b^{n-i}}-1$\\
			\midrule
			$i < n-1$ 	& 	\makecell[l]{append letter:\\ $q=q_w$,\\ $\delta(q_w,a)= q_{wa}$}
			& \makecell[l]{plug in vector:\\ $q=q_w$,\\ $\delta(q_w,a)= r_{f\circ\tau_{wa}}$}
			& \makecell[l]{plug in letter:\\ $q=r_g$,\\ $\delta(r_g,a)= r_{g\circ\tau_a}$}\\
			\bottomrule
			\end{tabular}	
			\caption{Defining $\delta(q,a)$ in terms of $i=i(q)$ for \Cref{thm:ub}.}\label{fig:define-delta}
			\end{figure}
				
			Since $f\in\mathfrak C$, we have
			\[
				\car{\{f\circ\tau_w:w\in [b]^j\}}\le \car{\{h\circ\tau_w: w\in [b]^j, h\in\mathfrak C\}}
			\]
			although this need not be strict
			(for instance, when $j=n$, we are comparing the range of $f$ to the union of ranges of $h$, $h\in\mathfrak C$, which may both equal $[c]$). 
			By construction, $M_f$ accepts $f$; see also \Cref{exa:1}, \Cref{exa:2}, and \Cref{exa:3}.
			\end{proof}
			\begin{example}\label{exa:1}
				The following example shows the case $b=c=2$ and $n=3$, with $f$ the majority function.
				It has
				$\chi(b,c,n)=1$:
				\[
					\xymatrix{
																& *+[Fo]{q_1}\ar[r]^{1}\ar[dr]_{0} 	& *+[Fo]{r_{\top_1}}\ar[dr]^{0,1}\\
					*+[Fo]{q_{\lam}}\ar[dr]_{0}\ar[ur]^{1} 	&									& *+[Fo]{r_{1_1}}\ar[r]_{1}		& *+[F]{r_{\top_0}} \\
																& *+[Fo]{q_0}\ar[ur]_{1}			&
					}
				\]
				The states $r_{g}$ for $g\in \mathfrak C^-_n\subseteq [c]\setminus\{0\}$ serve as our final states and are indicated by a rectangular box.
				Here $\top_k$ is the constant 1 function of $k$ variables, whereas $1_j$ is defined by $1_j(x)=1$ if $j=x$, 0 otherwise.
				There is no arrow labeled 0 between the states $q_0$ and $r_{1_0}$. This is because after seeing $x=y=0$ we already know the majority of $x,y,z$ is 0, so we ``reject by missing transition''.
			\end{example}
			\begin{example}\label{exa:2}
				A slightly larger example: the case $b=c=2$ and $n=4$, with $f$ the majority function.
				It has
				$\chi(b,c,n)=2$:
				\[
					\xymatrix{
																& *+[Fo]{q_1}\ar[r]^{1}\ar[dr]_{0} 	& *+[Fo]{q_{11}}\ar[r]^{0,1}	& *+[Fo]{r_{\top_1}}\ar[ddr]^{0,1}	&\\
																&									& *+[Fo]{q_{10}}\ar[ur]^1\ar@/_/[ddr]_<<<<0	&						&\\
					*+[Fo]{q_{\lam}}\ar[ddr]_{0}\ar[uur]^{1}	&									&					& 			& *+[F]{r_{\top_0}} \\
																&									& *+[Fo]{q_{01}}\ar@/_/[uuur]_1\ar[r]_-0	&*+[Fo]{r_{1_1}}\ar[ur]_1	&\\
																& *+[Fo]{q_0}\ar[ur]_{1}\ar[r]_{0}	& *+[Fo]{q_{00}}\ar[ur]_1	&
					}
				\]
				In this case, the upper bound is strict: $q_{01}$ and $q_{10}$ are equivalent.
				Thus a smaller automaton suffices:
				\[
					\xymatrix{
																& *+[Fo]{q_1}\ar[r]^{1}\ar[dr]_{0} 	& *+[Fo]{q_{11}}\ar[r]_{0,1}	& *+[Fo]{r_{\top_1}}\ar[dr]^{0,1}	&\\
					*+[Fo]{q_{\lam}}\ar[dr]_{0}\ar[ur]^{1}	&									& *+[Fo]{q_{01}}\ar[r]_0\ar[ur]_1	& *+[Fo]{r_{1_1}}\ar[r]_{1}			& *+[F]{r_{\top_0}} \\
																& *+[Fo]{q_0}\ar[ur]_{1}\ar[r]_{0}	& *+[Fo]{q_{00}}\ar[ur]_1	& 		&
					}
				\]
			\end{example}
			\begin{example}\label{exa:3}
				As an example for the case $c>2$, let $b=2$, $c=3$, $n=2$, and let $f(x,y)=x+y$. Then our automaton $M_f$ is:
				\[
					\xymatrix{
																& *+[Fo]{q_1}\ar[r]^{1}\ar[dr]_{0} 	& *+[F]{r_{2}}\\
					*+[Fo]{q_{\lam}}\ar[dr]_{0}\ar[ur]^{1} 	&									& *+[F]{r_{1}}\\
																& *+[Fo]{q_0}\ar[ur]_{1}			& 
					}
				\]
			\end{example}
		\Cref{thm:lb} is a generalization of C\^ampeanu and Ho's theorem.
		The construction is similar to that of~\cite[Figure 1 and Theorem 8]{MR2167778}.
		\begin{theorem}\label{thm:lb}
			Let $b, c\ge 2$ and $n\ge 1$ be integers.
			There exists a function $f:[b]^n\to [c]$ such that
			\(
				\sc(f)=\maxsc_{b,c,n}.
			\)
		\end{theorem}
		\begin{proof}
			Let $\mathfrak C=[c]^{[b]^n}$
			To define $f\in\mathfrak C$, we first note that it suffices to fix an $i$ with $0\le i\le n$ and define $f\circ \tau_w$ for each $w\in [b]^{i}$.
			To that end, we fix $i_0=\chi(b,c,n)$.
			Since
			\[
				\car{[b]^{i_0}} = b^{i_0}\ge c^{b^{n-{i_0}}}-1 = \car{\mathfrak C_{i_0}^-},
			\]
			there exists a surjective function $\phi:[b]^{i_0}\to\mathfrak C_{i_0}^-$.
			Define $f$ by $f\circ \tau_w = \phi(w)$ for each $w\in [b]^{i_0}$.
			We claim that $f$ attains the bound, i.e., there is no smaller automaton than that given in \Cref{thm:ub}.
			By \Cref{thm:shallit3.10.1}, an IDFA to accept $f$ is minimal if all states are reachable (from the start state) and any two states are $M$-distinguishable.

			Thus, it remains to show that the states for $f$ as given in the proof of \Cref{thm:ub} are reachable and $M$-distinguishable.
			
			By choice of $i_0$ it is easy to see that each state is reachable. For an example of what can go wrong with a different choice of $i_0$, see \Cref{fig:xor}.
			
			As for distinguishability, all states have a path to an accepting state,
			so it suffices to show that states that are the same
			distance from the start state are $M$-distinguishable.
			Recall that the set of states of $M_f$ is
			\[
				Q = \{q_w: w\in [b]^i, i\le i_0\} \cup
				\{r_{g}: g\in\mathfrak C^-_i, i>i_0\}
			\]
			For two states $q_v$, $q_w$ where $\abs{v}=\abs{w}$, it suffices to consider the case $\abs{w}=i_0$.
			Then $q_v$ and $q_w$ are $M$-distinguishable precisely because we chose $i_0$ and $f$ so that each extension by adding one more symbol to $v$, $w$
			does not give the same set of possible extensions, i.e., precisely to distinguish $v$ and $w$.
			Similarly $r_g$ and $r_h$ for $g,h\in\mathfrak C^-_i, i>i_0$ have the sets of possible extensions given by $g,h$ and therefore are $M$-distinguishable.
		\end{proof}
		\begin{figure}[h]
			\[
				\xymatrix{
						&\top_1\ar[dr]^{0,1}	&\\
					\mathrm{XOR}\ar[r]_{0}\ar[dr]_{1} & p\ar[r]_{1}		& \top_0\\
						&\neg p\ar[ur]_{0}	&\\
				}
			\]
			\caption{An unreachable state $\top_1$ in the automaton $M_{\mathrm{XOR},1}$ (\Cref{def:M}).}\label{fig:xor}
		\end{figure}

	\section{The number of maximally complex languages}\label{sec:number}
		A \emph{$k$-set} is a set of cardinality $k$.
		For a function $f:A\to B$ we denote the range and domain by $\ran(f)=\{f(x)\mid x\in A\}$ and $\dom(f)=A$, respectively.
		The collection of all subsets of $A$ of cardinality $k$ is denoted $\binom{A}k$.
		\begin{lemma}\label{lem:the-number}
			Let $k,b,v,i$ be positive integers with $i<v$.
			Let $Z:b\to v$ be the constant function defined by $Z(b')=v-1$ for all $b'\in [b]$.
			The number of $k$-sets $X\subseteq [v]^{[b]}\setminus\{Z\}$ such that $\bigcup_{f\in X}\ran(f)\supseteq [i]$ is
			\begin{equation}\label{eq:the-number}
				\binom{    v^b  -1}k -
				\sum_{j=1}^i
				(-1)^{j+1}
				\binom{i}{j}
					\binom{(v-j)^b-1}k.
			\end{equation}
		\end{lemma}
		\begin{proof}
			There are $v^b-1$ elements of $[v]^{[b]}\setminus\{Z\}$ and hence $\binom{v^b-1}k$ total $k$-sets.

			Since $i<v$, $v-1\not\in [i]=\{0,\dots,i-1\}$. Thus the range of $Z$ is disjoint from $[i]$.
			
			Given $J\subseteq [i]$, $\car{J}=j$,
			it follows that $Z\in ([v]\setminus J)^{[b]}$ and so there are $(v-j)^b-1$ functions in $[v]^{[b]}\setminus\{Z\}$ whose range is disjoint from $J$,
			i.e.,
			\[
				\car{\left([v]^{[b]}\setminus\{Z\}\right) \cap ([v]\setminus J)^{[b]}}=\car{([v]\setminus J)^{[b]}\setminus\{Z\}}=(v-j)^b-1.
			\]
			Here $([v]\setminus J)^{[b]}\subseteq [v]^{[b]}$.
			
			For the union of ranges to not contain $[i]$ means that there is some $i'\in [i]$ that is missed.
			The number of $k$-sets that miss some $i'$ is then given by inclusion-exclusion in terms of $j$,
			the cardinality of a set $J\subseteq [i]$ that is disjoint from $\bigcup_{f\in X}\ran(f)$.
			Thus the number of $k$-sets with $\bigcup_{f\in X}\ran(f)\not\supseteq [i]$ is
			\[
				\sum_{j=1}^i
				(-1)^{j+1}
				\binom{i}{j}
					\binom{(v-j)^b-1}k.
				\qedhere
			\]
		\end{proof}
		For fixed $b$ and $c$, let $\mathcal B_d$ ($\mathcal B_d^+$) be the set of all (not constant zero) functions from $[b]^d$ to $[c]$.
		\begin{definition}\label{def:chi}
		For a function $f:[b]^n\to [c]$ and $0\le j\le n$, define a function
		\(
			\varphi_{f,j}: [b]^j \to [c]^{[b]^{n-j}}
		\)
		by
		\(
			\varphi_{f,j}(w) = f\circ\tau_{w}
		\) for all $w$.
		\end{definition}
		Note that $\varphi_{f,j}$ is the function $\phi$ in the proof of \Cref{thm:lb}.
		\begin{definition}\label{def:adequate}
			For each $0\le j\le n$, let $Z_j:[b]^{n-j}\to [c]$ be the constant zero function.
			A set $X\subseteq [c]^{[b]^{n-j}}\setminus\{Z_j\}$ is \emph{$j$-adequate} if
			\[
				\{f\circ\tau_{a}\mid f\in X, a\in [b]\}\supseteq [c]^{[b]^{n-(j+1)}}\setminus\{Z_{j+1}\}.
			\]
			A function $\varphi: [b]^j \to [c]^{[b]^{n-j}}$ is called \emph{$j$-adequate} if
			its range is a $j$-adequate $b^j$-set, i.e.:
				\begin{enumerate}
					\item $\varphi(w)\ne Z_{j}$ for each $w$,
					\item $\varphi$ is injective, and 
					\item 
					\(
						\{(\varphi(w))\circ\tau_{a}\mid w\in [b]^{j}, a\in [b]\}\supseteq [c]^{[b]^{n-(j+1)}}\setminus\{Z_{j+1}\}.
					\)
				\end{enumerate}
			We say that $\varphi$ is \emph{adequate} if it is $j$-adequate for $j=\chi(b,c,n)$.
		\end{definition}
		\begin{proposition}\label{prop:j-adequate}
			If $\varphi$ is $j$-adequate then $b^j\le c^{b^{n-j}}-1$ and $b^{j+1}\ge c^{b^{n-(j+1)}}-1$.
		\end{proposition}
		The proof of \Cref{prop:j-adequate} is immediate. It follows that $\varphi$ can only be $j$-adequate if $j=\chi(b,c,n)$,
		unless we happen to have $b^{j+1}=c^{b^{n-(j+1)}}-1$.
		\begin{proposition}\label{prop:determine}
			For all $j$, we have $f=g\iff \varphi_{f,j}=\varphi_{g,j}$.
		\end{proposition}
		\begin{proof}
			$\implies$ is immediate. Conversely, suppose $\varphi_{f,j}=\varphi_{g,j}$. Fix $x$ and write $x=x_1x_2$, $\abs{x_1}=j$. Then 
			\[
				f(x)=\varphi_{f,j}(x_1)(x_2)=\varphi_{g,j}(x_1)(x_2)=g(x).\qedhere
			\]
		\end{proof}
		\begin{definition}\label{def:M}
		For each $f$ and $j$ we defined the associated automaton $M_{f,j}$ in \Cref{3pm-121721}.
		Let $M^-_{f,j}$ be $M_{f,j}$ with unreachable states removed and indistinguishable states merged.
		Let $Q^-_{f,j}$ be the set of states of $M^-_{f,j}$.
		\end{definition}
		\begin{theorem}\label{thm:jun19 2021}
			The following are equivalent:
			\begin{enumerate}
				\item $\varphi_{f,\chi(b,c,n)}$ is adequate.
				\item $\#(Q^-_{f,\chi(b,c,n)})=\maxsc_{b,c,n}$; all states of $Q^-_{f,\chi(b,c,n)}$ are reachable and distinguishable; and $M^-_{f,\chi(b,c,n)}$ accepts $f$.
				\item It is not the case that:
				\(
				\#(Q^-_{f,\chi(b,c,n)})<\maxsc_{b,c,n},
				\)
				and $M^-_{f,\chi(b,c,n)}$ accepts $f$.
			\end{enumerate}
		\end{theorem}
		\begin{proof}
			(2) $\implies$ (1): If $\varphi_f$ is not adequate then by definition some states of $M_{f,\chi(b,c,n)}$ are not reachable.

			(1) $\implies$ (2): \Cref{thm:lb}.

			(2) $\implies$ (3) is immediate.

			(3) $\implies$ (2): Assume (3). Since $M^-_{f,\chi(b,c,n)}$ always accept $f$, it follows that it has $\ge\maxsc_{b,c,n}$ states.
			By \Cref{thm:ub} it has exactly $\maxsc_{b,c,n}$ states.
		\end{proof}
		\begin{theorem}\label{prop:adequate}
			The following are equivalent:
			\begin{enumerate}
				\item $\varphi_{f,\chi(b,c,n)}$ is adequate.
				\item $\sc(f)=\maxsc_{b,c,n}$.
			\end{enumerate}
		\end{theorem}
		\begin{proof}
			(1) $\implies$ (2): by (1) $\implies$ (2) of \Cref{thm:jun19 2021} and then by \Cref{thm:blue}.
			
			(2) $\implies$ (1): Suppose $\neg$(1). Then $\neg$(1) in \Cref{thm:jun19 2021}. Therefore $\neg$(3) in \Cref{thm:jun19 2021},
			and so $\sc(f)<\maxsc_{b,c,n}$.
		\end{proof}

			\newcommand{\nmcformula}{				k!\mleft(
				\binom{c^{b^{d+1}}-1}k
				-
				\sum_{j=1}^i
				(-1)^{j+1}
				\binom{i}{j}
					\binom{(c^{b^{d}}-j)^{b}-1}k
				\mright)
}
		\begin{proposition}\label{prop:binomial}
			Let $b,c,n$ be given, $i_0=\chi(b,c,n)$, $i=c^{b^{n-(i_0+1)}}-1$, and $k=b^{i_0}$.
			The number of adequate functions $\varphi:[b]^{i_0}\to [c]^{[b]^{n-i_0}}$ is
			\begin{equation}\label{eq:nmcf}
				\nmcformula.
			\end{equation}
		\end{proposition}
		\begin{proof}
			If $\alpha_i$ is the number of adequate sets then
			the number of adequate functions $\varphi$ is $k!\,\alpha_i$. 

			The map $\varphi$ maps to functions whose union of ranges covers the next set of functions as in \Cref{lem:the-number},
			$k$-sets $X=\{f_1,\dots,f_k\}\subseteq [v]^{[b]}\setminus\{Z\}$ such that $\bigcup_{f\in X}\ran(f)\supseteq [i]$ where
			$i=c^{b^{n-(i_0+1)}}-1$.

			Let $Z_0:[b]^{n-i_0}\to [c]$ be the constant zero function.
			Let $Z(a)=c^{b^{n-(i_0+1)}}-1$ for all $a$.
			Let
			\[
				\beta: [c]^{[b]^{n-i_0}} \to {[c^{b^{n-(i_0+1)}}]}^{[b]}
			\]
			be an arbitrary bijection for which $\beta(Z_0)=Z$.
			By \Cref{lem:the-number}, applying $\beta$,
			and with $v=c^{b^d}$,
			\[
				\alpha_i =
				\binom{c^{b^{d+1}}-1}k
				-
				\sum_{j=1}^i
				(-1)^{j+1}
				\binom{i}{j}
					\binom{(c^{b^{d}}-j)^{b}-1}k.
			\]
			Thus, the number of maps $\varphi$ is
			\begin{eqnarray*}
				&& k!\mleft(
				\binom{c^{b^{d+1}}-1}k
				-
				\sum_{j=1}^i
				(-1)^{j+1}
				\binom{i}{j}
					\binom{(c^{b^{d}}-j)^{b}-1}k
				\mright).\qedhere
			\end{eqnarray*}
		\end{proof}
		\begin{theorem}\label{thm:jan16-2021}
			Let integers $b,c\ge 2$ and $n\ge 1$ be given.
			Let $k=b^{j_0}$ where $j_0=\chi(b,c,n)$.
			Let $d+1=n-j_0$ and $i=c^{b^{d}}-1$.
			Then $\car{\{f\mid \sc(f)=\maxsc_{b,c,n}\}}$ is given by~\eqref{eq:nmcf} and equals
			\[
			\nmcformula.
			\]
		\end{theorem}
		\begin{proof}
			By \Cref{prop:determine},
			\[
				\#\{f:\sc(f)=\maxsc_{b,c,n}\}=\#\{\varphi_f:\sc(f)=\maxsc_{b,c,n}\}.
			\]
			By \Cref{prop:adequate} this equals $\#\{\varphi_f: f\text{ is adequate}\}$, which
			by \Cref{prop:binomial} equals~\eqref{eq:nmcf}.
		\end{proof}

		\begin{example}\label{January 14 2021}
			For $n=4$ and $b=c=2$, we have $\chi(b,c,n)=2$, as illustrated in the following table:

			\begin{center}
				\begin{tabular}{r| r r r r r}
					$i$ &	0&	1&	2&	3&	4\\
					\hline
					$2^i$&	1&	2&	4&	8&	16\\
					$\car{\mathcal B_{n-i}^+}=2^{2^{n-i}}-1$ &65535& 255&15&3&1\\
					\hline
					$\min(2^i,2^{2^{n-i}}-1)$ & 1& 2&	4&	3&1
				\end{tabular}
			\end{center}

			A maximal complexity function is determined by an injective function $\phi$ from $[2]^2$ to $\mathcal B_2^+$,
			such that $\bigcup \{\ran\phi(w)\mid w\in [2]^2\}\supseteq\mathcal B_1^+$.
			Associating each $\phi$ with the set $X_{\phi}=\{\phi(w)\mid w\in [2]^2\}\in\binom{\mathcal B_2^+}4$,
			we see that the number of functions $\phi$ is $4!$ times the number of four-element subsets $X$ of $\mathcal B_2^+$
			for which $\bigcup\{\ran f\mid f\in X\}\supseteq\mathcal B_1^+$.
			By \Cref{lem:the-number} that number is 1155: let $b=2$, $k=4$, $i=3$, and $v=2^2=4$ and calculate that~\eqref{eq:the-number} is $\binom{15}4-3\binom{8}4=1155$.
			Thus the total number of maximum complexity functions is $1155\cdot 24 = 27720$.
		\end{example}

 		\section{Asymptotics}\label{sub:sur}
		In this section we demonstrate (\Cref{thm:sur}) that while most functions do not have maximum complexity, the growth rate of the number of maximally complex functions
		is similar to that of the total number of function $f:[b]^n\to [c]$ for $b=c=2$.
		\begin{proposition}\label{thm:jun14-2021}
			Suppose $i\le v$ and $k$ are positive integers, and $A$ is a set.
			Suppose $(k-1)\car{A}<i$. Then we have
				\begin{eqnarray}
					&&\car{\{\varphi:[k]\to [v]^{A}\mid \bigcup_{t\in [k]} \ran(\varphi(t)) \supseteq [i] \}}\label{eq:not1:1} \\
					&=&
					k! \cdot \car{\{X\in \binom{[v]^{A}}k \mid \bigcup_{f\in X} \ran(f) \supseteq [i] \}}.\label{eq:adequate}
				\end{eqnarray}
				\item
				Suppose that additionally $i<v$, and $Z:A\to [v]$ is a constant function with $Z(a)=z\not\in [i]$ for all $a\in A$.
				Then~\eqref{eq:not1:1} also equals
				\begin{eqnarray}
					k! \cdot \car{\left\{X\in \binom{[v]^{A}}k \mid Z\not\in X, \bigcup_{f\in X} \ran(f) \supseteq [i] \right\}}.\label{eq:adequate2}
				\end{eqnarray}
		\end{proposition}
		\begin{proof}
			\eqref{eq:not1:1}=\eqref{eq:adequate}:
			Let $\varphi$ be given and let $X=\ran(\varphi)$. It suffices to show that $\car{X}=k$.
			Since
			\[
				\car{\ran(f)}\le\car{\dom(f)}=\car{A}
			\] for each $f\in X$, we have
			\[
				i \le \car{\bigcup_{t\in [k]} \ran(\varphi(t))}
				    = \car{
					\bigcup \{
						\ran(f)\mid f\in X
					\}
					}\le \car{X}\car{A}.
			\]
			If $\car{X}\ne k$ then $\car{X}\le k-1$, and we have the contradiction
			\begin{equation}\label{eq:contra}
				i\le \car{X}\car{A}\le (k-1)\car{A} < i.
			\end{equation}
			\eqref{eq:not1:1}=\eqref{eq:adequate2}:
			When $Z$ is constant equal to a value not in $[i]$, $Z\not\in\ran(\varphi)$ follows from the other condition:
			if $Z\in\ran(\varphi)$ then let $X=\ran(\varphi)\setminus\{Z\}$. Then $\car{X}=k-1$
			and we get a contradiction as in~\eqref{eq:contra}.
		\end{proof}

		\begin{definition}\label{forCatalan}
			Let $b$ and $c$ be positive integers and let $0\le t\le n$.
			Let $O_t=O_t^{(b,c,n)}$ be the number of functions from $[b^{t}]$ to $[c^{b^{n-t}}]$ that are onto $[c^{b^{n-t}}-1]$:
			\[
				\forall y\in [c^{b^{n-t}}-1]\quad\exists x\in [b^t]\quad f(x)=y.
			\]
		\end{definition}
		\begin{theorem}\label{thm:sur}
			Let $b$ and $c$ be positive integers and let $n\ge 0$.
			Let $j_0=\chi(b,c,n)$.
			If the condition
			\begin{equation}\label{eq:sur-condition}
			b\cdot (b^{j_0}-1) < c^{b^{n-(j_0+1)}}-1
			\end{equation}
			holds, then $\car{\{f:[b]^n\to [c]\mid \sc(f)=\maxsc_{b,c,n}\}}=O_t$, where $0\le t\le n$ is minimal such that $O_t>0$.
		\end{theorem}
		\begin{proof}
			By \Cref{prop:adequate},
			\[
				\car{\{f\mid \sc(f)=\maxsc_{b,c,n}\}}=\car{\{f\mid \varphi_{f,\chi(b,c,n)}\text{ is adequate}\}}.
			\]
			Let $Z_0:[b]^{n-j_0}\to [c]$ be the constant zero function.
			Let $Z(a)=c^{b^{n-(j_0+1)}}-1$ for all $a$.
			Let
			\[
				\beta: [c]^{[b]^{n-j_0}} \to {[c^{b^{n-(j_0+1)}}]}^{[b]}
			\]
			be an arbitrary bijection for which $\beta(Z_0)=Z$.

			Given $\varphi$ define $\psi$ by $\psi(x,b')=\varphi(x)\circ\tau_{b'}$.
			The following are equivalent:
			\begin{itemize}
				\item $\bigcup_{x\in [k]}\ran(\varphi(x))\supseteq [i]$;
				\item $\psi$ is onto $[i]$.
			\end{itemize}
			Thus $O_t$ is equal to~\eqref{eq:not1:1}, where $t=\chi(b,c,n)+1$.
			By \Cref{thm:jun14-2021}
			under the bijection $\beta$, with
			$k=b^{j_0}$,
			$A=[b]$,
			$i=c^{b^{n-j_0}}-1$, and
			$v=c^{b^{n-j_0}}$,
			$O_t$ is moreover equal to~\eqref{eq:adequate2}, as desired.
		\end{proof}

		\begin{remark}
			The authors regret that in~\cite{MR3943984}, the condition~\eqref{eq:sur-condition} in \Cref{thm:sur} was erroneously omitted.
			By definition $b^{j_0+1}>c^{b^{n-(j_0+1)}}-1$, so $b^{j_0+1}\ge c^{b^{n-(j_0+1)}}$, but the condition fails when $b^{j_0+1}\ge c^{b^{n-(j_0+1)}}+b-1$.
		\end{remark}
		\begin{example}
			Consider the case $n=1$, $b=c=2$ of \Cref{thm:sur}. Then $i_0=1$, where $i_0$ is the least $i$ such that $O_i>0$.
			$O_i^{(2,2,1)}$ is the number of functions from $[2^i]$ to $[2^{2^{1-i}}]$ that are onto $[2^{2^{1-i}}-1]$.
			For $i=0$, there are no such functions.
			For $i=1$, there are three such functions.
			And indeed, this is the number of maximal complexity functions in this case: the functions $f:\{0,1\}\to \{0,1\}$ that are onto $\{1\}$.
		\end{example}
		\begin{definition}
			Let $O_{m,n}$ be the number of onto functions from $[m]$ to $[n]$.
			Stirling numbers of the second kind are denoted ${m\brace n}$ and equal the number of equivalence relations on $[m]$ with $n$ equivalence classes.
		\end{definition}
		The following result is well known.
		\begin{lemma}\label{lem:wellKnown}
			Let $m,n$ be positive integers. Then
			\(
				O_{m,n} = n!{m\brace n}.
			\)
		\end{lemma}

		\begin{lemma}\label{lem:onto-some}
			Let $u$ and $v$ be positive integers.
			The number of functions from $[u]$ to $[v]$ that are onto the first $v-1$ elements of $[v]$ is
			\[
				(v-1)!\sum_{m=0}^{u-(v-1)} \binom{u}{m} {u-m\brace v-1}.
			\]
			The number of functions from $[u]$ to $[v]$ that are onto $[i]$ is
			\[
				i!\sum_{m=0}^{u-i} \binom{u}{m} {u-m\brace i} (v-i)^m.
			\]
		\end{lemma}
		\begin{proof}
			Let $m$ be the number of elements going to $[v]\setminus [i]$.
			Then we see that the number of such functions is
			\[
				\sum_{m=0}^{u-i} \binom{u}{m} O_{u-m,i} (v-i)^{m}= \sum_{m=0}^{u-i} \binom{u}{m} i!{u-m\brace i} (v-i)^m
			\]
			by \Cref{lem:wellKnown}.
		\end{proof}
		\begin{lemma}\label{lem:factorial}
			Let $u$ be a positive integer.
			The number of functions from $[u]$ to $[u]$ that are onto $[u-1]$ is $(u+1)!/2$.
		\end{lemma}
		\begin{proof}
			Note that for any $m$, ${m\brace m}=1$ and ${m\brace m-1}={m\choose 2}$.
			By \Cref{lem:onto-some}, the number of such functions is
		\begin{eqnarray*}
			&& (u-1)!\sum_{m=0}^{1} \binom{u}{m} {u-m\brace u-1}\\
			&=& (u-1)!\mleft(\binom{u}{0} {u\brace u-1} + \binom{u}{1} {u-1\brace u-1}\mright)\\
			&=& (u-1)!\mleft( {u\brace u-1} + u {u-1\brace u-1}\mright)\\
			&=& (u-1)!\mleft( {u\choose 2} + u \mright)
			= \frac{(u+1)!}2.\qedhere
		\end{eqnarray*}
		\end{proof}
		The following \Cref{lem:gamma} will only be applied in the case $\gamma=0$.
		\begin{lemma}\label{lem:gamma}
			Let $j$ be a nonnegative integer, let $p\ge 2$, and let $0\le\gamma\le j$ be an integer.
			Let $n=p^j+j-\gamma$ and $b=p$, $c=p^{p^\gamma}$.
			$O_i:=O^{b,c,n}_i$, where $i$ is minimal such that $O_i>0$, equals
			\[
				 \frac{(p^{p^j}+1)!}2.
			\]
		\end{lemma}
		\begin{proof}
			The condition that $O_i^{(b,c,n)}>0$ for some $i$, i.e., $b^i\ge c^{b^{n-i}}-1$ for some $0\le i\le n$,
			i.e., $p^i\ge p^{p^\gamma p^{n-i}}-1$, i.e., $p^n\ge p^{p^\gamma}-1$, i.e., either $n\ge p^\gamma$ (i.e., $\gamma\le j$) or $n=\gamma=0, p=2$,
			follows from $\gamma\le j$.

			By \Cref{lem:onto-some}, with $u=b^i, v=c^{b^{n-i}}$,
			\begin{eqnarray*}
				O^{(b,c,n)}_i = (v-1)!\sum_{m=0}^{u-(v-1)} \binom{u}{m} {u-m\brace v-1}.
			\end{eqnarray*}

			The condition $O_i>0$ is equivalent to $b^i\ge c^{b^{n-i}}-1$.
			When $b=c=p$ and $i>0$, this is equivalent to
			\begin{equation}\label{eq:gamma}
				i\ge p^\gamma p^{n-i}
			\end{equation}
			Let $k=p^{j}$.
			Since by assumption $n=k+j-\gamma$, \eqref{eq:gamma} becomes
			\[
				ip^{i} \ge p^\gamma p^n = kp^k
			\]
			Since the map $i\mapsto i p^i$ is increasing, the requirement for $i$ is that $i\ge k$.
			Note that setting $i=k$ now makes $u=v$.
			Therefore by \Cref{lem:factorial}, $O_i$ is $(u+1)!/2=(p^i+1)!/2$ as desired.
		\end{proof}

		\begin{lemma}\label{lem:the-number-sur}
			Let $j$ be a nonnegative integer.
			Let $n=p^j+j$ and $b=c=p \ge 2$.
			Then
			\[
				\car{\{f:[b]^n\to [c]\mid \sc(f)=\maxsc_{b,c,n}\}}= \frac{(p^{p^j}+1)!}2.
			\]
		\end{lemma}
		\begin{proof}
			We have
			$p^{n-j}=p^{p^j}$ and hence $p^{m}=p^{p^{n-m}}$ for $m=n-j$,
			so that $p^m > p^{p^{n-m}}-1$ but $p(p^{m-1}-1)< p^{p^{n-m}}-1$.
			Thus \Cref{thm:sur} applies and the number of such functions is
			$O_i:=O^{b,c,n}_i$, where $i$ is minimal such that $O_i>0$.
			By \Cref{lem:gamma} with $\gamma=0$ we are done.
		\end{proof}
		Using \Cref{thm:jan16-2021} for $b=c=2$ we calculate some values for
		\[
			\mathrm{nmcf}(n) = \#\{f:[2]^n\to [2]\mid \sc(f)=\maxsc_{2,2,n}\},
		\]
		the number of maximally complex functions from $[2]^n$ to $[2]$, in \Cref{tab:nmcf}.
		In \Cref{thm:nmcf} we shall study the limiting behavior suggested by \Cref{tab:nmcf}.
		\begin{table}
			\centering
			\begin{tabular}{r r r}
				\toprule
				$n$ & $\mathrm{nmcf}(n)$ & $\frac1n\log_2\log_2(\mathrm{nmcf}(n))$\\
				\midrule
				0 & 1 & \\
				$1$ &                   $3$& 0.664\\
				$2$ &                   $6$& 0.685\\
				$3$ &                  $60$& 0.854\\
				$4$ &               $27720$& 0.971\\
				$5$ &           $259338240$& 0.961\\
				$6$ &     $177843714048000$& 0.927\\
				\bottomrule
			\end{tabular}
			\caption{The number of maximally complex functions from $[2]^n$ to $[2]$ for $n\le 6$.}\label{tab:nmcf}
		\end{table}
		\begin{theorem}\label{thm:nmcf}
			The number of maximal complexity functions satisfies
			\[
				\limsup_{n\to\infty}\frac1n\log_2\log_2(\mathrm{nmcf}(n))=1.
			\]
		\end{theorem}
		\begin{proof}
			It is immediate that
			\(
				\limsup_{n\to\infty}\frac1n\log_2\log_2(\mathrm{nmcf}(n))\le 1
			\).
			For the other direction, consider the case where $n=2^j+j$ for some $j$.
			By Stirling's approximation,
			\begin{eqnarray*}
				&&\log_2 (2^{2^j}!) = 2^{2^j}\log_2 2^{2^j} - 2^{2^j} \log_2 e +O(\log_2 2^{2^j})\\
				&=& 2^{2^j} 2^j - 2^{2^j} \log_2 e +O(2^j)\\
				&=& 2^n - 2^{n-j} \log_2 e +O(2^j)
			\end{eqnarray*}
			and hence
			\(
				\lim_{n\to\infty}\frac1n\log_2\log_2 (2^{2^{j}}!)=1
			\).
			By \Cref{lem:the-number-sur},
			\[
				\limsup_{n\to\infty}\frac1n\log_2\log_2(\mathrm{nmcf}(n))\ge 1.\qedhere
			\]
		\end{proof}

		In \Cref{lem:the-number-sur}, $(p^{p^j}+1)!/2$ may seem like a large number but it is relatively small:
		in terms of $w:=p^k$,
		\begin{eqnarray*}
			\frac{(p^{p^j}+1)!/2}{c^{b^n}}&=&
			\frac{(p^{p^j}+1)!/2}{p^{p^\gamma p^n}}=\\
			\frac{(p^{p^j}+1)!/2}{p^{p^{p^j+j}}} &=& \frac{(p^k+1)!/2}{p^{p^{k+j}}} = \frac{(p^k+1)!/2}{p^{k p^k}} = \frac{(w+1)!/2}{w^w}\to 0.
		\end{eqnarray*}
		\begin{example}
		For $n=6$ and $b=c=2$, then, we get
		$i=k=2^j=4$, and
		\[
			O_4^{2,2,6} = \frac{17!}2
		\]
		So there are more than 177 trillion maximum-complexity 6-ary Boolean functions, which is however a small fraction of the total number of such functions,
		\[
			2^{2^6} = 18,446,744,073,709,551,616
		\]
		or over 18 quintillion.
		\end{example}
		\begin{figure}\label{tab5}
			\begin{align*}
				&\{000,001,010,101\}, \{000,001,010,111\}, \{000,001,011,100\},\\
				&\{000,001,100,111\}, \{000,001,011,110\}, \{000,001,101,110\},\\
				&\{000,010,011,101\}, \{000,010,011,111\}, \{001,010,011,100\},\\
				&\{010,011,100,111\}, \{001,010,011,110\}, \{010,011,101,110\},\\
				&\{000,011,100,101\}, \{000,100,101,111\}, \{001,010,100,101\},\\
				&\{010,100,101,111\}, \{001,100,101,110\}, \{011,100,101,110\},\\
				&\{000,011,110,111\}, \{000,101,110,111\}, \{001,010,110,111\},\\
				&\{010,101,110,111\}, \{001,100,110,111\}, \{011,100,110,111\},\\
				&\{000,001,010,100,111\},\{000,001,010,101,110\},\{000,001,011,100,110\},\\
				&\{000,001,011,101,110\},\{000,001,011,100,111\},\{000,001,010,101,111\},\\
				&\{000,010,011,100,111\},\{000,010,011,101,110\},\{001,010,011,100,110\},\\
				&\{001,010,011,101,110\},\{001,010,011,100,111\},\{000,010,011,101,111\},\\
				&\{000,010,100,101,111\},\{000,011,100,101,110\},\{001,010,100,101,110\},\\
				&\{001,011,100,101,110\},\{001,010,100,101,111\},\{000,011,100,101,111\},\\
				&\{000,010,101,110,111\},\{000,011,100,110,111\},\{001,010,100,110,111\},\\
				&\{001,011,100,110,111\},\{001,010,101,110,111\},\{000,011,101,110,111\},\\
				&\{000,001,010,011,100,111\}, \{000,001,010,011,101,110\},\\
				&\{000,001,010,100,101,111\}, \{000,001,011,100,101,110\},\\
				&\{000,001,010,101,110,111\}, \{000,001,011,100,110,111\},\\
				&\{000,010,011,100,101,111\}, \{001,010,011,100,101,110\},\\
				&\{000,010,011,101,110,111\}, \{001,010,011,100,110,111\},\\
				&\{000,011,100,101,110,111\}, \{001,010,100,101,110,111\}.
			\end{align*}
			\caption{
				The $(2^{2^1}+1)!/2=60$ languages $Z=\{x\mid f(x)=1\}$ with maximal complexity, 7,
				for $n=3=2^1+1$ and $b=c=2$.}
		\end{figure}

		\begin{remark}
			For future work, it would be interesting (but difficult) to determine the distribution of $\sc(f)$ over $f\in [c]^{[b]^n}$.
		\end{remark}
		\newpage
		\bibliographystyle{plain}
		\bibliography{kh-liu}
\end{document}